\documentclass[11pt]{article} 

\usepackage{geometry}                		
\usepackage{graphicx}				

\usepackage{array}								
\usepackage{amssymb}
\usepackage{amsmath}
\usepackage{mathtools}
\usepackage{braket}
\usepackage{amsthm}
\usepackage{dsfont}							
\usepackage{verbatim}						
\usepackage{MnSymbol}						
\usepackage{mathrsfs}						
\usepackage[dvipsnames]{xcolor}
\usepackage{authblk}						

\newtheorem{theorem}{Theorem}

\newtheorem{lemma}{Lemma}[section]
\newtheorem{corollary}{Corollary}

\newcommand\numberthis{\addtocounter{equation}{1}\tag{\theequation}}

  {\renewcommand{\thecorollary}{\ref{#1}$'$}%
   \addtocounter{corollary}{-1}%
   \begin{corollary}}
  {\end{corollary}}
  
\newcommand{\rndBrk}[1]{\left( #1 \right)}
\newcommand{\sqBrk}[1]{\left[ #1 \right]}
\newcommand{\curlyBrk}[1]{\left\lbrace #1 \right\rbrace}
\newcommand{\norm}[1]{\left\Vert #1 \right\Vert}

\DeclareMathOperator{\tr}{Tr}

\DeclareMathOperator{\Id}{\mathds{1}}

\begin{document}


\title{Uniform continuity bound for sandwiched R\'enyi conditional entropy}
\author{Ashutosh Marwah and Fr\'ed\'eric Dupuis}
\affil{\small{D\'epartement d'informatique et de recherche op\'erationnelle,\\ Universit\'e de Montr\'eal,\\ Montr\'eal QC, Canada}}


\date{\today}

\maketitle 

\begin{abstract}
    We prove a simple uniform continuity bound for the sandwiched R\'enyi conditional entropy for $\alpha \in [1/2, 1) \cup (1, \infty]$, which is independent of the dimension of the conditioning system. 
\end{abstract}

\section{Introduction}

Conditional entropies quantify the amount of uncertainty of a system when one has access to side information. R\'enyi conditional entropies have been instrumental in understanding both classical and quantum information theoretic processes in the one-shot domain~\cite{Arikan96, Renner05,Hayashi13, Dupuis20, Friedman18, Dupuis21}. There are multiple generalizations and definitions of R\'enyi conditional entropies for quantum states. One of the most important of these is the sandwiched R\'enyi conditional entropy~\cite{Lennert13,Wilde14} (which we denote as $\tilde{H}_{\alpha}^{\uparrow} (A|B)_{\rho}$ following the notation in Tomamichel's book, Quantum Information Processing with Finite Resources~\cite{TomamichelBook16}), which is defined for a quantum state $\rho_{AB}$, $\alpha \in [1/2, 1) \cup (1, \infty]$ and $\alpha' = (\alpha-1)/\alpha$ as 
\begin{align*}
    \tilde{H}_{\alpha}^{\uparrow} (A|B)_{\rho} &= \sup_{\eta_B} \frac{1}{1-\alpha} \log \tr \rndBrk{\eta_B^{-\alpha'/2} \rho_{AB} \eta_B^{-\alpha'/2}}^{\alpha}.
\end{align*}
where the supremum is taken over all quantum states $\eta_B$. For $\alpha = 1/2$, this definition results in $H_{\max}$, which characterizes the rate of compression for a probabilistic source~\cite{Konig09} and for $\alpha = \infty$, it results in $H_{\min}$, which has been used to characterize the amount of randomness which one can extract from a system independent of the conditioning system~\cite{Renner05}. Recently, the sandwiched conditional entropies with $\alpha \in (1, 2]$ have also been used to derive bounds for randomness extraction~\cite{Dupuis21}. \\

Many applications require that we are able to bound the conditional entropy of one state in terms of another state, which is close to the original state. For von Neumann conditional entropy, the Alicki-Fannes-Winter (AFW) uniform continuity bound~\cite{Alicki04,Winter16} enables us to do this. It states that for two quantum states $\rho_{AB}$ and $\sigma_{AB}$ such that $\frac{1}{2}\norm{\rho_{AB} - \sigma_{AB}}_1 \leq \epsilon$, we have 
\begin{align}
    \vert H(A|B)_{\rho} - H(A|B)_{\sigma}\vert \leq 2 \epsilon \log(d_A) + (1+\epsilon)h\left(\frac{\epsilon}{1+\epsilon} \right)
    \label{eq:OrgAFW}
\end{align}
where the function $h(x) = -x \log(x)- (1- x) \log(1-x)$ is the binary entropy and $d_A$ is the dimension of the vector space $A$. A key feature of this bound is that it is independent of the size $d_B$ of the conditioning system. This is crucial in many applications, particularly in cryptography, where the conditioning system often represents an adversary's system whose size cannot be bounded. The AFW bound has been used widely in quantum information theory (see, for example, the converse bounds for channel capacities~\cite{Wilde13}). \\

In this paper, we will prove a uniform continuity bound, like the AFW bound, for sandwiched R\'enyi conditional entropies. Specifically, in Theorem \ref{th:AFW_for_Renyi_leq_1} and Corollary \ref{cor:AFW_for_Renyi_geq_1}, we prove that for two quantum states $\rho_{AB}$ and $\sigma_{AB}$ such that $\frac{1}{2}\norm{\rho_{AB} - \sigma_{AB}}_1 \leq \epsilon$, and for $\alpha < 1$ we have
\begin{align*}
    \vert \tilde{H}_{\alpha}^{\uparrow} (A|B)_{\sigma} - \tilde{H}_{\alpha}^{\uparrow} (A|B)_{\rho} \vert \leq \log(1+\epsilon) + \frac{1}{1-\alpha} \log \rndBrk{1 + \epsilon^{\alpha} d_A^{2(1-\alpha)} - \frac{\epsilon}{(1+\epsilon)^{1-\alpha}} }
\end{align*}
and for $\alpha>1$, we have
\begin{align*}
    \vert \tilde{H}_{\alpha}^{\uparrow} (A|B)_{\sigma} &- \tilde{H}_{\alpha}^{\uparrow} (A|B)_{\rho} \vert \\
        &\leq \log(1+\sqrt{2\epsilon}) + \frac{1}{1-\beta} \log \rndBrk{1 + (\sqrt{2\epsilon})^{\beta} d_A^{2(1-\beta)} - \frac{\sqrt{2\epsilon}}{(1+\sqrt{2\epsilon})^{1-\beta}}} 
\end{align*}
where $\beta$ is such that $\alpha^{-1}+ \beta^{-1}=2$.\\

To prove these, we make use of a simple lemma about the subadditivity of the function $X \mapsto \tr(X^\alpha)$ for $\alpha<1$ and use the duality relation for sandwiched R\'enyi conditional entropy, which states that for a pure tripartite state $\rho_{ABC}$ and an $\alpha \in [\frac{1}{2}, 1) \cup (1, \infty]$, $\tilde{H}_\alpha^{\uparrow}(A|B)_{\rho} = -\tilde{H}_\beta^{\uparrow}(A|C)_{\rho}$ for $\beta$ such that $\alpha^{-1} + \beta^{-1} = 2$.\\

Jabbour and Datta~\cite{Jabbour20} proved the following tight uniform continuity bound for classical and quantum-classical states when $\alpha<1$:\footnote{Jabbour and Datta~\cite{Jabbour20} use the optimized Petz R\'enyi conditional entropy definition as the definition for the quantum R\'enyi conditional entropy. For classical states and quantum-classical states, though, their definition of conditional entropy is equal to the one considered in this paper.}
\begin{align}
    \vert \tilde{H}_{\alpha}^{\uparrow} (A|B)_{\sigma} - \tilde{H}_{\alpha}^{\uparrow} (A|B)_{\rho} \vert \leq \frac{1}{1-\alpha} \log \rndBrk{(1-\epsilon)^\alpha + \epsilon^\alpha (d_A-1)^{1-\alpha}}.
    \label{eq:Jabbour20_AFW_alpha_leq_1_cl}
\end{align}
The proof for Eq.~\eqref{eq:Jabbour20_AFW_alpha_leq_1_cl} uses conditional majorization~\cite{Gour18} and a series of transforms to tightly convert the probability distributions in question into distributions it is easier to argue about. The continuity bound proven in this paper for the R\'enyi conditional entropy of classical probability distributions (see Eq.~\eqref{eq:AFW_alpha_leq_1_cl} in Theorem~\ref{th:AFW_for_Renyi_leq_1}) is almost the same as Eq.~\eqref{eq:Jabbour20_AFW_alpha_leq_1_cl}. In addition, the proof of our main theorem (Theorem~\ref{th:AFW_for_Renyi_leq_1} below) is much simpler than the proof of Eq.~\eqref{eq:Jabbour20_AFW_alpha_leq_1_cl}. Also Leditzky et al.~\cite{Leditzky16} proved using H\"older's inequality that for $\alpha \in [\frac{1}{2}, 1)$
\begin{align*}
   \tilde{H}_{\alpha}^{\uparrow} (A|B)_{\rho} - \tilde{H}_{\beta}^{\uparrow} (A|B)_{\sigma} \geq \frac{2\alpha}{1-\alpha} \log F(\rho_{AB}, \sigma_{AB})
\end{align*}
where $\beta$ is such that $\alpha^{-1}+ \beta^{-1}=2$ and $F$ is the fidelity function. Although this inequality is not a continuity bound (since the Rényi parameter is different in the two terms), it can often be used as a continuity bound in applications where one of the states is particularly simple (e.g.~maximally entangled, or maximally mixed). (Also see \cite{Wang19,Eisert22} for related inequalities.)

\section{Uniform continuity bound for quantum R\'enyi entropies}

In the following, the dimension of a vector space $\mathcal{V}$ will be succinctly referred to as $d_{\mathcal{V}}$. Also, for Hermitian matrices $P$ and $Q$, we will use the notation $P \geq Q$ to denote that the matrix $P-Q$ is positive semidefinite. \\

We will use the sandwiched R\'enyi relative entropy~\cite{Lennert13,Wilde14} for our proof. For $\alpha \in [\frac{1}{2}, 1) \cup (1, \infty]$ and positive operators $P$ and $Q$ such that $\tr(P)\neq 0$, the sandwiched R\'enyi relative entropy of $P$ relative to $Q$ is defined as
\begin{align*}
    \tilde{D}_\alpha (P || Q) := 
    \begin{cases} 
        \frac{1}{\alpha-1} \log \frac{\tr \rndBrk{Q^{-\alpha'/2} P Q^{-\alpha'/2}}^\alpha}{\tr(P)} &\text{ if } (\alpha<1 \text{ and } P \nperp Q) \text{ or } \ker(Q) \subseteq \ker (P)\\
        \infty & \text{else.}
    \end{cases}
\end{align*}
For $\alpha \in [\frac{1}{2}, 1) \cup (1, \infty]$, the sandwiched R\'enyi relative entropy satisfies the data processing inequality, that is, for any completely positive and trace-preserving map $\Phi$ and positive operators $P$ and $Q$, we have 
\begin{align*}
    \tilde{D}_\alpha (P || Q) \geq \tilde{D}_\alpha (\Phi(P) || \Phi(Q)).
\end{align*}
We begin by stating a lemma about the subadditivity of the trace of matrices raised to the power of $\alpha \in [0,1]$. This was proven by McCarthy~\cite{McCarthy67}. We provide a proof here for completeness. 

\begin{lemma}[McCarthy's Inequality~\cite{McCarthy67}]
	For positive semidefinite matrices $P$ and $Q$ and $\alpha \in [0,1]$, we have 
    \begin{align*}
        \tr((P+Q)^{\alpha}) \leq \tr(P^{\alpha}) + \tr(Q^{\alpha})
    \end{align*}
    \label{lemm:alphaSubadditivity}
\end{lemma}
\begin{proof}
    First, note that the statement is trivial for $\alpha = 0$ and $\alpha = 1$. For the other values of $\alpha$, we will use the following integral representation of $t^{\alpha}$ where $t \in [0, \infty)$ and $\alpha \in (0,1)$ (see for example Example V.1.10 in Bhatia's Matrix Analysis~\cite{Bhatia97})
    \begin{align*}
        t^{\alpha} = \int_{0}^{\infty} \frac{t}{\lambda + t} d\mu(\lambda)
    \end{align*}
    where $\mu$ is a positive measure on $(0, \infty)$. For a positive semidefinite matrix $X$, and $\alpha \in (0,1)$ we can use this representation to write 
    \begin{align*}
        X^{\alpha} = \int_{0}^{\infty} {X}{(\lambda + X)^{-1}} d\mu(\lambda).
    \end{align*}
    Thus, we have 
    \begin{align*}
        \tr((P+Q)^{\alpha}) &= \int_{0}^{\infty} \tr \rndBrk{{(P+Q)}{(\lambda + P+Q)^{-1}}} d\mu(\lambda) \\
        & = \int_{0}^{\infty} \tr \rndBrk{{P}{(\lambda + P+Q)^{-1}}} d\mu(\lambda) + \int_{0}^{\infty} \tr \rndBrk{{Q}{(\lambda + P+Q)^{-1}}} d\mu(\lambda) \\
        & \leq \int_{0}^{\infty} \tr \rndBrk{{P}{(\lambda + P)^{-1}}} d\mu(\lambda) + \int_{0}^{\infty} \tr \rndBrk{{Q}{(\lambda + Q)^{-1}}} d\mu(\lambda)\\
        & = \tr(P^{\alpha} + Q^{\alpha})
    \end{align*}
    where for the inequality we use the fact that the function $X \mapsto X^{-1}$ is operator anti-monotone (see for example Proposition V.1.6 in Bhatia's Matrix Analysis~\cite{Bhatia97}).
\end{proof}

We will now use the lemma above to prove the following uniform continuity bound for sandwiched R\'enyi entropies with $\alpha < 1$. 

\begin{theorem}
    For $\alpha < 1$, $\epsilon \in [0,1]$ and bipartite normalized quantum states $\rho_{AB}$ and $\sigma_{AB}$ such that $\frac{1}{2}\norm{\rho_{AB} - \sigma_{AB}}_1 \leq \epsilon$, the difference in sandwiched R\'enyi conditional entropy for the two states can be bounded as 
    \begin{align}
        \vert \tilde{H}_{\alpha}^{\uparrow} (A|B)_{\sigma} - \tilde{H}_{\alpha}^{\uparrow} (A|B)_{\rho} \vert \leq \log(1+\epsilon) + \frac{1}{1-\alpha} \log \rndBrk{1 + \epsilon^{\alpha} d_A^{2(1-\alpha)} - \frac{\epsilon}{(1+\epsilon)^{1-\alpha}} }.
        \label{eq:AFW_alpha_leq_1}
    \end{align}
    If, in addition, system $A$ is classical for both $\rho_{AB}$ and $\sigma_{AB}$ or system $B$ is classical for both states, then we can strengthen the bound to 
    \begin{align}
        \vert \tilde{H}_{\alpha}^{\uparrow} (A|B)_{\sigma} - \tilde{H}_{\alpha}^{\uparrow} (A|B)_{\rho} \vert \leq \log(1+\epsilon) + \frac{1}{1-\alpha} \log \rndBrk{1 + \epsilon^{\alpha} d_A^{1-\alpha} - \frac{\epsilon}{\rndBrk{d_A(1+\epsilon)}^{1-\alpha}} }.
        \label{eq:AFW_alpha_leq_1_cl}
    \end{align}
    \label{th:AFW_for_Renyi_leq_1}
\end{theorem}

\begin{proof}
    For $\alpha<1$, we can write the sandwiched R\'enyi conditional entropy as 
    \begin{align}
        2^{(1-\alpha)\tilde{H}_{\alpha}^{\uparrow} (A|B)_{\rho}} = \sup_{\eta_B} \tr \rndBrk{\eta_B^{-\alpha'/2} \rho_{AB} \eta_B^{-\alpha'/2}}^{\alpha}.
        \label{eq:H_alpha_exp_form}
    \end{align}
    This form will turn out to be useful later.\\

    We can assume that $\epsilon>0$ as otherwise the bound is trivial. Since the upper bound in the theorem is increasing in the trace distance, it is sufficient to prove the bound for $\norm{\rho_{AB} - \sigma_{AB}}_1 = 2\epsilon$. Let $\rho_{AB} - \sigma_{AB} = P'_{AB} - Q'_{AB}$ be the decomposition of $\rho_{AB} - \sigma_{AB}$ into its positive and negative parts, so that $P'_{AB} \geq 0$ and $Q'_{AB} \geq 0$. Note that $\tr(P'_{AB}) = \tr(Q'_{AB}) = \epsilon$. Thus, we can further define the density operators $P_{AB} := P'_{AB}/ \epsilon$ and $Q_{AB} := Q'_{AB}/ \epsilon$. \\

    Now, we have that $\sigma_{AB} + \epsilon P_{AB} = \rho_{AB} + \epsilon Q_{AB}$. For any positive semidefinite operator $\eta_B$, we have the following chain of implications
    \begin{align*}
        & \sigma_{AB} + \epsilon P_{AB} = \rho_{AB} + \epsilon Q_{AB} \\
        \Rightarrow &\ \eta_B^{-\alpha'/2} \rndBrk{\sigma_{AB} + \epsilon P_{AB}}  \eta_B^{-\alpha'/2} = \eta_B^{-\alpha'/2} \rho_{AB}\ \eta_B^{-\alpha'/2} + \epsilon \eta_B^{-\alpha'/2} Q_{AB}\ \eta_B^{-\alpha'/2} \\
        \Rightarrow & \tr \rndBrk{\eta_B^{-\alpha'/2} \rndBrk{\sigma_{AB} + \epsilon P_{AB}} \eta_B^{-\alpha'/2}}^\alpha = \tr \rndBrk{\eta_B^{-\alpha'/2} \rho_{AB}\ \eta_B^{-\alpha'/2}+ \epsilon \eta_B^{-\alpha'/2} Q_{AB}\ \eta_B^{-\alpha'/2}}^\alpha \\
        \Rightarrow & \tr \rndBrk{\eta_B^{-\alpha'/2} \rndBrk{\sigma_{AB} + \epsilon P_{AB}} \eta_B^{-\alpha'/2}}^\alpha \leq \tr \rndBrk{\eta_B^{-\alpha'/2}\rho_{AB}\eta_B^{-\alpha'/2}}^\alpha + \epsilon^\alpha \tr \rndBrk{ \eta_B^{-\alpha'/2} Q_{AB} \eta_B^{-\alpha'/2}}^\alpha
    \end{align*}
    where we use Lemma \ref{lemm:alphaSubadditivity} in the third step. Taking the supremum on the right-hand side, we get
    \begin{align*}
        \tr \rndBrk{\eta_B^{-\alpha'/2} \rndBrk{\sigma_{AB} + \epsilon P_{AB}} \eta_B^{-\alpha'/2}}^\alpha & \leq \sup_{\omega_B} \curlyBrk{ \tr \rndBrk{\omega_B^{-\alpha'/2}\rho_{AB}\omega_B^{-\alpha'/2}}^\alpha + \epsilon^\alpha \tr \rndBrk{ \omega_B^{-\alpha'/2} Q_{AB} \omega_B^{-\alpha'/2}}^\alpha } \\
        & \leq \sup_{\omega_B} \curlyBrk{ \tr \rndBrk{\omega_B^{-\alpha'/2}\rho_{AB}\omega_B^{-\alpha'/2}}^\alpha } + \epsilon^\alpha \sup_{\omega_B} \curlyBrk{\tr \rndBrk{ \omega_B^{-\alpha'/2} Q_{AB} \omega_B^{-\alpha'/2}}^\alpha }.
    \end{align*}
    Using Eq.~\eqref{eq:H_alpha_exp_form}, we derive that for every state $\eta_B$
    \begin{align*}
        \tr \rndBrk{\eta_B^{-\alpha'/2} \rndBrk{\sigma_{AB} + \epsilon P_{AB}} \eta_B^{-\alpha'/2}}^\alpha &\leq 2^{(1-\alpha)\tilde{H}_{\alpha}^{\uparrow} (A|B)_{\rho}} + \epsilon^\alpha 2^{(1-\alpha)\tilde{H}_{\alpha}^{\uparrow} (A|B)_{Q}}\\
        & \leq 2^{(1-\alpha)\tilde{H}_{\alpha}^{\uparrow} (A|B)_{\rho}} + \epsilon^\alpha d_A^{1-\alpha}
        \numberthis \label{eq:intermediate_upper_bound}
    \end{align*}
    where we have used $\tilde{H}_{\alpha}^{\uparrow} (A|B)_{Q} \leq \log(d_A)$ to derive the last line. Now, we will try to lower bound the term on the left-hand side. For an arbitrary $\delta>0$, there exist states $\mu_B$ and $\nu_B$ such that
    \begin{align}
        2^{(1-\alpha)\tilde{H}_{\alpha}^{\uparrow} (A|B)_{\sigma}} &\leq \tr \rndBrk{\mu_B^{-\alpha'/2} \sigma_{AB} \mu_B^{-\alpha'/2}}^{\alpha} +\delta \label{eq:def_mu}\\
        2^{(1-\alpha)\tilde{H}_{\alpha}^{\uparrow} (A|B)_{P}} &\leq \tr \rndBrk{\nu_B^{-\alpha'/2} P_{AB} \nu_B^{-\alpha'/2}}^{\alpha} +\delta . \label{eq:def_nu}
    \end{align}
    We use these states to define the state $\theta_{BC}$ as
    \begin{align*}
        \theta_{BC} := \frac{1}{1+\epsilon} \mu_B \otimes \ket{0}\bra{0}_C + \frac{\epsilon}{1+\epsilon} \nu_B \otimes \ket{1}\bra{1}_C.
    \end{align*} 
    Also, define 
    \begin{align*}
        \bar{\sigma}_{ABC} := \frac{1}{1+\epsilon} \sigma_{AB} \otimes \ket{0}\bra{0}_C + \frac{\epsilon}{1+\epsilon} P_{AB} \otimes \ket{1}\bra{1}_C.
    \end{align*}
    With these definitions, we have the following chain of inequalities
    \begin{align*}
        \tr \left( \theta_B^{\frac{-\alpha'}{2}} (\sigma_{AB} + \epsilon P_{AB}) \theta_B^{\frac{-\alpha'}{2}} \right)^\alpha &= (1+\epsilon)^\alpha \tr \left( \theta_B^{\frac{-\alpha'}{2}} \bar{\sigma}_{AB} \theta_B^{\frac{-\alpha'}{2}} \right)^{\alpha} \\
        &= (1+\epsilon)^\alpha 2^{-(1-\alpha) \tilde{D}_\alpha(\bar{\sigma}_{AB}|| \Id_A \otimes \theta_B)}\\ 
        &\geq (1+\epsilon)^\alpha 2^{-(1-\alpha) \tilde{D}_\alpha(\bar{\sigma}_{ABC}|| \Id_A \otimes \theta_{BC})}\\ 
        &= (1+\epsilon)^{\alpha-1} \tr \rndBrk{\mu_B^{-\alpha'/2} \sigma_{AB} \mu_B^{-\alpha'/2}}^{\alpha} + \\
        & \hspace{2.5 cm} \epsilon(1+\epsilon)^{\alpha-1} \tr \rndBrk{\nu_B^{-\alpha'/2} P_{AB} \nu_B^{-\alpha'/2}}^{\alpha} \\
        &\geq (1+\epsilon)^{\alpha-1} \sqBrk{2^{(1-\alpha) \tilde{H}_{\alpha}^{\uparrow}(A|B)_{\sigma}} + \epsilon 2^{(1-\alpha) \tilde{H}_{\alpha}^{\uparrow}(A|B)_{P}}} \\
        & \hspace{2.5 cm} - (1+\epsilon)^{\alpha} \delta \\
        &\geq (1+\epsilon)^{\alpha-1} \left[ 2^{(1-\alpha) \tilde{H}_{\alpha}^{\uparrow}(A|B)_{\sigma}} + \epsilon d_A^{\alpha-1} \right]\\
        & \hspace{2.5 cm} - (1+\epsilon)^{\alpha} \delta 
    \end{align*}
    where we use the data processing inequality for $\tilde{D}_\alpha$ for the first inequality, the definitions of the states $\mu_B$ and $\nu_B$ (Eq.~\eqref{eq:def_mu}, \eqref{eq:def_nu}) in the second last step and the fact that $\tilde{H}_{\alpha}^{\uparrow}(A|B)_{P} \geq -\log(d_A)$ in the last step. Combining this inequality with the upper bound in Eq.~\eqref{eq:intermediate_upper_bound} yields
    \begin{align*}
        2^{(1-\alpha) \tilde{H}_{\alpha}^{\uparrow}(A|B)_{\sigma}} &\leq (1+\epsilon)^{1-\alpha} 2^{(1-\alpha) \tilde{H}_{\alpha}^{\uparrow}(A|B)_{\rho}} + (1+\epsilon)^{1-\alpha} \epsilon^{\alpha} d_A^{1-\alpha} - \epsilon d_A^{\alpha-1} + (1+\epsilon)\delta.
    \end{align*}
    Since $\delta$ is arbitrary we can let $\delta \rightarrow 0$. Further, taking the logarithm on both sides and dividing by $1-\alpha >0$, we get
    \begin{align*}
        \tilde{H}_{\alpha}^{\uparrow} (A|B)_{\sigma} &\leq \frac{1}{1-\alpha} \log \rndBrk{(1+\epsilon)^{1-\alpha} 2^{(1-\alpha) \tilde{H}_{\alpha}^{\uparrow}(A|B)_{\rho}} + (1+\epsilon)^{1-\alpha} \epsilon^{\alpha} d_A^{1-\alpha} - \epsilon d_A^{\alpha-1}}
    \end{align*}
    We can now bound the difference of the conditional entropies as 
    \begin{align*}
        \tilde{H}_{\alpha}^{\uparrow} (A|B)_{\sigma} - \tilde{H}_{\alpha}^{\uparrow} (A|B)_{\rho} &\leq \frac{1}{1-\alpha} \log \rndBrk{(1+\epsilon)^{1-\alpha} 2^{(1-\alpha) \tilde{H}_{\alpha}^{\uparrow}(A|B)_{\rho}} + (1+\epsilon)^{1-\alpha} \epsilon^{\alpha} d_A^{1-\alpha} - \epsilon d_A^{\alpha-1}} \\
        & \hspace{2.5 cm} -\tilde{H}_{\alpha}^{\uparrow} (A|B)_{\rho} \\
        &= \frac{1}{1-\alpha} \log \rndBrk{(1+\epsilon)^{1-\alpha} + \rndBrk{(1+\epsilon)^{1-\alpha} \epsilon^{\alpha} d_A^{1-\alpha} - \epsilon d_A^{\alpha-1}}2^{-(1-\alpha) \tilde{H}_{\alpha}^{\uparrow}(A|B)_{\rho}}}\\
        &\leq \frac{1}{1-\alpha} \log \rndBrk{(1+\epsilon)^{1-\alpha} + \rndBrk{(1+\epsilon)^{1-\alpha} \epsilon^{\alpha} d_A^{1-\alpha} - \epsilon d_A^{\alpha-1}}d_A^{1-\alpha}} \\
        &= \log(1+\epsilon) + \frac{1}{1-\alpha} \log \rndBrk{1 + \epsilon^{\alpha} d_A^{2(1-\alpha)} - \frac{\epsilon}{(1+\epsilon)^{1-\alpha}} }
    \end{align*}
    where we used the fact that $(1+\epsilon)^{1-\alpha} \epsilon^{\alpha} d_A^{1-\alpha} \geq \epsilon d_A^{\alpha-1}$ since $d_A \geq 1$, which implies that the right-hand side in the second line is a decreasing function in $\tilde{H}_{\alpha}^{\uparrow} (A|B)_{\rho}$, which we again lower bound by $-\log(d_A)$ (when $A$ or $B$ is classical, we can use $\tilde{H}_{\alpha}^{\uparrow} (A|B)_{\rho} \geq 0$ and derive the stronger bound in Eq.~\eqref{eq:AFW_alpha_leq_1_cl}). Since we did not use any specific property of $\rho_{AB}$ or $\sigma_{AB}$, one can derive the same upper bound by switching $\rho_{AB}$ and $\sigma_{AB}$.
\end{proof}

\noindent It is worth noting that in the limit $\alpha \rightarrow 1$, the bound in Eq.~\eqref{eq:AFW_alpha_leq_1} yields the bound
\begin{align*}
    \vert H(A|B)_{\rho} - H(A|B)_{\sigma} \vert \leq 2 \epsilon \log(d_A) + (1+ \epsilon) \log(1+ \epsilon) - \epsilon \log(\epsilon)
\end{align*}
which is equal to the tight continuity bound for conditional von Neumann entropy (Eq. \ref{eq:OrgAFW}) proven by Winter \cite{Winter16}. This indicates that our bound is close to being tight in the regime where $\alpha$ is close to $1$ and $\alpha<1$.\\

We will now prove a uniform continuity bound for sandwiched R\'enyi conditional entropies with $\alpha>1$ as a corollary to the above theorem by making use of the duality relation for these entropies. Recall that according to the duality relation~\cite{Lennert13,Beigi13} for $\tilde{H}_{\alpha}^{\uparrow} (A|B)_{\rho}$, we have that for a pure state $\rho_{ABC}$
\begin{align*}
    \tilde{H}_{\alpha}^{\uparrow} (A|B)_{\rho} = - \tilde{H}_{\beta}^{\uparrow} (A|C)_{\rho} 
\end{align*}
for $\alpha, \beta \in [1/2,1) \cup (1, \infty]$ satisfying $\alpha^{-1} + \beta^{-1} = 2$. Note that for $\alpha \in (1, \infty]$, the $\beta$ satisfying the duality condition lies in $[1/2, 1)$.\\

\begin{corollary}
    For $\alpha > 1$, $\epsilon \in [0,1]$ and bipartite normalized quantum states $\rho_{AB}$ and $\sigma_{AB}$ such that $\frac{1}{2}\norm{\rho_{AB} - \sigma_{AB}}_1 \leq \epsilon$, the difference in sandwiched R\'enyi conditional entropy for the two states can be bounded as 
    \begin{align*}
        \vert \tilde{H}_{\alpha}^{\uparrow} (A|B)_{\sigma} &- \tilde{H}_{\alpha}^{\uparrow} (A|B)_{\rho} \vert \\
        &\leq \log(1+\sqrt{2\epsilon}) + \frac{1}{1-\beta} \log \rndBrk{1 + (\sqrt{2\epsilon})^{\beta} d_A^{2(1-\beta)} - \frac{\sqrt{2\epsilon}}{(1+\sqrt{2\epsilon})^{1-\beta}}} 
        \numberthis
        \label{eq:AFW_alpha_geq_1}
    \end{align*}
    for $\beta= \frac{\alpha}{2\alpha-1}$ defined such that $\alpha^{-1} + \beta^{-1} = 2$.
    \label{cor:AFW_for_Renyi_geq_1}
\end{corollary}
\begin{proof}
    If $\frac{1}{2} \norm{\rho_{AB} - \sigma_{AB}}_1 \leq \epsilon$, then using the Fuchs-van de Graaf inequality along with Uhlmann's theorem, we can find purifications $\rho_{ABC}$ and $\sigma_{ABC}$ such that 
    \begin{align*}
        \frac{1}{2}\norm{\rho_{ABC} - \sigma_{ABC}}_1 \leq \sqrt{2\epsilon}.
    \end{align*}
    For $\alpha>1$, choose $\beta <1 $ such that $\alpha^{-1} + \beta^{-1} = 2$, then using the duality relation and Theorem \ref{th:AFW_for_Renyi_leq_1} we have that 
    \begin{align*}
        \vert \tilde{H}_{\alpha}^{\uparrow} (A|B)_{\sigma} - \tilde{H}_{\alpha}^{\uparrow} (A|B)_{\rho} \vert &= \vert \tilde{H}_{\beta}^{\uparrow} (A|C)_{\sigma} - \tilde{H}_{\beta}^{\uparrow} (A|C)_{\rho} \vert \\
        &\leq \log(1+\sqrt{2\epsilon}) + \frac{1}{1-\beta} \log \rndBrk{1 + (\sqrt{2\epsilon})^{\beta} d_A^{2(1-\beta)} - \frac{\sqrt{2\epsilon}}{(1+\sqrt{2\epsilon})^{1-\beta}} }
    \end{align*}
\end{proof}
We expect that the bound in Eq.~\eqref{eq:AFW_alpha_geq_1} is not tight in the constants and the exponents. For $H_{\min} = \tilde{H}_{\infty}^{\uparrow}$, the corollary above gives us the bound
\begin{align*}
    \vert {H}_{\min} (A|B)_{\sigma} - {H}_{\min}  (A|B)_{\rho} \vert \leq \log(1+\sqrt{2\epsilon}) + 2 \log \rndBrk{1 + (2\epsilon)^{1/4} d_A - \frac{\sqrt{2\epsilon}}{(1+\sqrt{2\epsilon})^{1/2}}} 
\end{align*}
It is easy to prove a tighter uniform continuity bound with an upper bound of $\log(1+\epsilon d_A^2)$ for $H_{\min}$. We prove it in the Appendix as Theorem \ref{th:AFW_min_entropy}. Nevertheless, the bound above is independent of the conditioning system, and it should be sufficient for most applications. 

\section{Conclusion}

We proved a simple uniform continuity bound for the sandwiched R\'enyi conditional entropy. Further work should look at improving these bounds, especially for $\alpha>1$. Deriving continuity bounds for other R\'enyi conditional entropies also remains an open problem. 

\section*{Acknowledgements}
We would like to thank Mark Wilde for his comments and for pointing us to work related to this paper. AM is supported by bourse d'excellence Google and would like to thank Google for their financial support. This research was also supported by Canada's NSERC.

\appendix

\section{Uniform continuity bound for $H_{\min}$}
\begin{sloppypar}
In the following theorem, we prove a simple and tighter uniform continuity bound for $H_{\min} = \tilde{H}_{\infty}^{\uparrow}$. Recall that $H_{\min}$ can be equivalently defined as
\begin{align*}
    H_{\min}(A|B)_{\rho} := \max \curlyBrk{\lambda \in \mathbb{R} : \text{there exists $\eta_B$ such that } \rho_{AB} \leq 2^{- \lambda} \Id_A \otimes \eta_B}.
\end{align*}
\begin{theorem}
    For $\epsilon \in [0,1]$ and bipartite normalized quantum states $\rho_{AB}$ and $\sigma_{AB}$ such that $\frac{1}{2}\norm{\rho_{AB} - \sigma_{AB}}_1 \leq \epsilon$, the difference in the min-entropy for the two states can be bounded as 
    \begin{align}
        \vert H_{\min} (A|B)_{\sigma} - H_{\min} (A|B)_{\rho} \vert \leq \log(1+\epsilon d_A^2).
        \label{eq:AFW_min_entropy}
    \end{align}
    \label{th:AFW_min_entropy}
\end{theorem}
\begin{proof}
    Once again, since the upper bound in the theorem is increasing in $\epsilon$, we can simply prove the bound for states $\rho_{AB}$ and $\sigma_{AB}$ such that $\frac{1}{2}\norm{\rho_{AB} - \sigma_{AB}}_1 = \epsilon$. We also assume that $\epsilon>0$. Then, following the argument in Theorem \ref{th:AFW_for_Renyi_leq_1}, we can write $\rho_{AB} - \sigma_{AB} = \epsilon P_{AB} - \epsilon Q_{AB}$ for density operators $P_{AB}$ and $Q_{AB}$. We will now use the fact that $\sigma_{AB} \leq \rho_{AB} + \epsilon Q_{AB}$. Using the definition of $H_{\min}(A|B)_\rho$, we know that there exists a state $\eta_B$ such that
    \begin{align*}
        \rho_{AB} \leq 2^{-H_{\min}(A|B)_\rho} \Id_A \otimes \eta_B.
    \end{align*}
    Also, since $H_{\min}(A|B)_Q \geq -\log(d_A)$, we have
    \begin{align*}
        Q_{AB} & \leq 2^{-H_{\min}(A|B)_Q} \Id_A \otimes \eta'_B \\
        & \leq d_A \Id_A \otimes \eta'_B
    \end{align*}
    for some state $\eta'_B$. Thus, we have
    \begin{align*}
        \sigma_{AB} & \leq \rho_{AB} + \epsilon Q_{AB} \\
        & \leq 2^{-H_{\min}(A|B)_\rho} \Id_A \otimes \eta_B + \epsilon d_A \Id_A \otimes \eta'_B \\
        & \leq (2^{-H_{\min}(A|B)_\rho} + \epsilon d_A) \Id_A \otimes \rndBrk{\frac{2^{-H_{\min}(A|B)_\rho}  \eta_B + \epsilon d_A \eta'_B}{2^{-H_{\min}(A|B)_\rho} + \epsilon d_A}}
    \end{align*}
    which implies that 
    \begin{align*}
        & 2^{-H_{\min}(A|B)_{\sigma}} \leq 2^{-H_{\min}(A|B)_{\rho}} + \epsilon d_A\\
        &\Rightarrow 2^{H_{\min}(A|B)_{\rho} - H_{\min}(A|B)_{\sigma}} \leq 1 + \epsilon d_A 2^{H_{\min}(A|B)_{\rho}}\\
        &\Rightarrow H_{\min}(A|B)_{\rho} - H_{\min}(A|B)_{\sigma} \leq \log\left( 1 + \epsilon d_A^2 \right).
    \end{align*}
    A similar bound can be proven in the opposite direction if we follow the above procedure with $\rho_{AB}$ and $\sigma_{AB}$ interchanged.
\end{proof}
\end{sloppypar}

\bibliographystyle{unsrt}
\bibliography{bib}

\end{document}